\documentclass[a4paper,12pt]{llncs}
\usepackage{amsmath,amssymb}
\usepackage{hyperref}
\usepackage[left=3.2cm,top=3.2cm,right=3.2cm,bottom=3.2cm]{geometry}
\usepackage{epsfig}

\usepackage{color}

\usepackage{graphicx}
\usepackage{subfigure}
\usepackage{floatrow}
\newfloatcommand{capbtabbox}{table}[][\FBwidth]

\usepackage{wrapfig}
\usepackage{amssymb}
\usepackage{amsmath}
\usepackage{tikz, hyperref}

\usepackage{longtable}
\usepackage{multirow}
\usepackage{array}
\usepackage{paralist}

\DeclareMathOperator{\diam}{diam}
\newcommand{\cone}{\textsc{cone}}

\newcommand{\FRUSTUM}{\textsc{frus}}

\newcommand{\vol}{\textsc{vol}}

\date{}
\begin{document}
\title{\Large{Iterative models for complex networks formed by extending cliques}\thanks{Research supported by grants from NSERC and the Fields Insitute.}}
\author{Anthony Bonato\inst{1}, Ryan Cushman\inst{1}, Trent G.\ Marbach\inst{1}, \and Zhiyuan Zhang \inst{1}}
\institute{Toronto Metropolitan University}
\maketitle
\begin{abstract}

We consider a new model for complex networks whose underlying mechanism is extending dense subgraphs. In the frustum model, we iteratively extend cliques over discrete-time steps. For many choices of the underlying parameters, graphs generated by the model densify over time. In the special case of the cone model, generated graphs provably satisfy properties observed in real-world complex networks such as the small world property and bad spectral expansion. We finish with a set of open problems and next steps for the frustum model. 

\end{abstract}

\section{Introduction}

The vast volume of data mined from the web and other networks from the physical, natural, and social sciences, suggests a view of many real-world networks as self-organizing phenomena satisfying common properties. Such complex networks capture dyadic interactions in many phenomena, ranging from friendship ties in Facebook, to Bitcoin transactions, to interactions between proteins in living cells. Complex networks evolve via a number of mechanisms such as preferential attachment or copying that predict how links between vertices are
formed over time. Key empirically observed properties of complex networks are the small world property (which predicts small distances between typical pairs of vertices and high local clustering), power law degree distributions (where most vertices have low degree, but there are a small number of high degree vertices), and densification (where the average degree tends to infinity with time). Early models such as preferential attachment \cite{ba,bol} successfully captured these properties and others. See the book \cite{bbook} for a survey of early complex networks models, along with \cite{at}.

Cliques are simplified representations of highly interconnected structures in networks. For example, in the Facebook social network, a clique consists of accounts linked via friendship or mutual interests. Cliques are of interest in network science as one type of \emph{motif}, which are certain small-order significant subgraphs; this higher-order network perspective has lead to a focus on hypergraphs in network science \cite{ak,be}. We may view the hypergraph of cliques in a network as a kind of backbone, which allows for the rapid diffusal of information and influence to nodes over time. Cliques in social and other networks grow organically, and therefore it is natural to consider models simulating their evolution. 

We introduce the \emph{frustum model}, which is simplified model for clique evolution. This elementary-seeming model leads to rich dynamics, generating graphs sharing many of the properties observed in complex networks. As an illustrative instance of the frustum model, whose precise definition will appear in Section~2, consider the \emph{cone model}. If in the $(t-1)$th time-step the model generated a graph $G_{t-1}$, then in the $t$th time-step and for each existing vertex $u$ in $G_{t-1}$, a clique of a prescribed order is added that is adjacent to $u.$ An illustration of several time-steps of the cone model is given in Figure~\ref{figcone}.

\begin{figure}[h]
\begin{minipage}{.04\textwidth}\centering
\includegraphics[scale=.005]{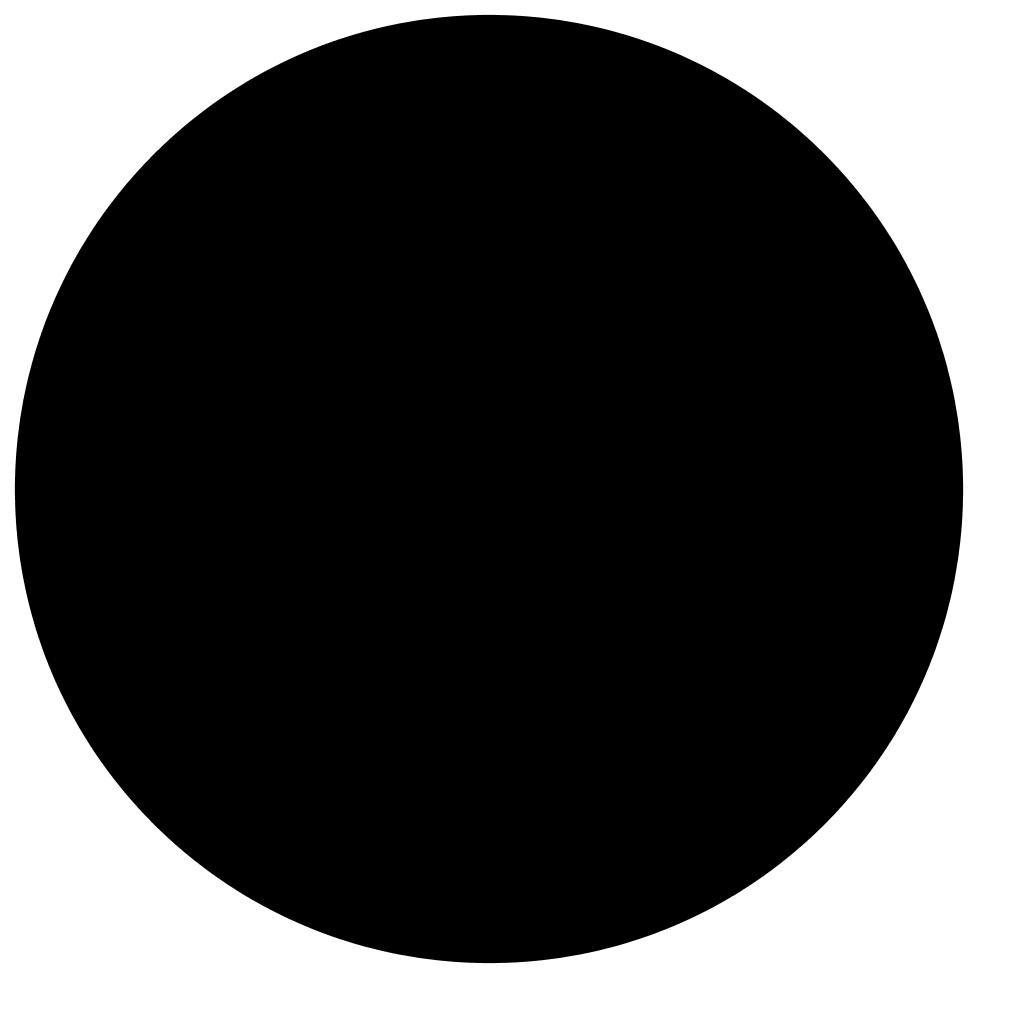}
\end{minipage}
\begin{minipage}{.07\textwidth}\centering
\includegraphics[scale=.03, angle=30]{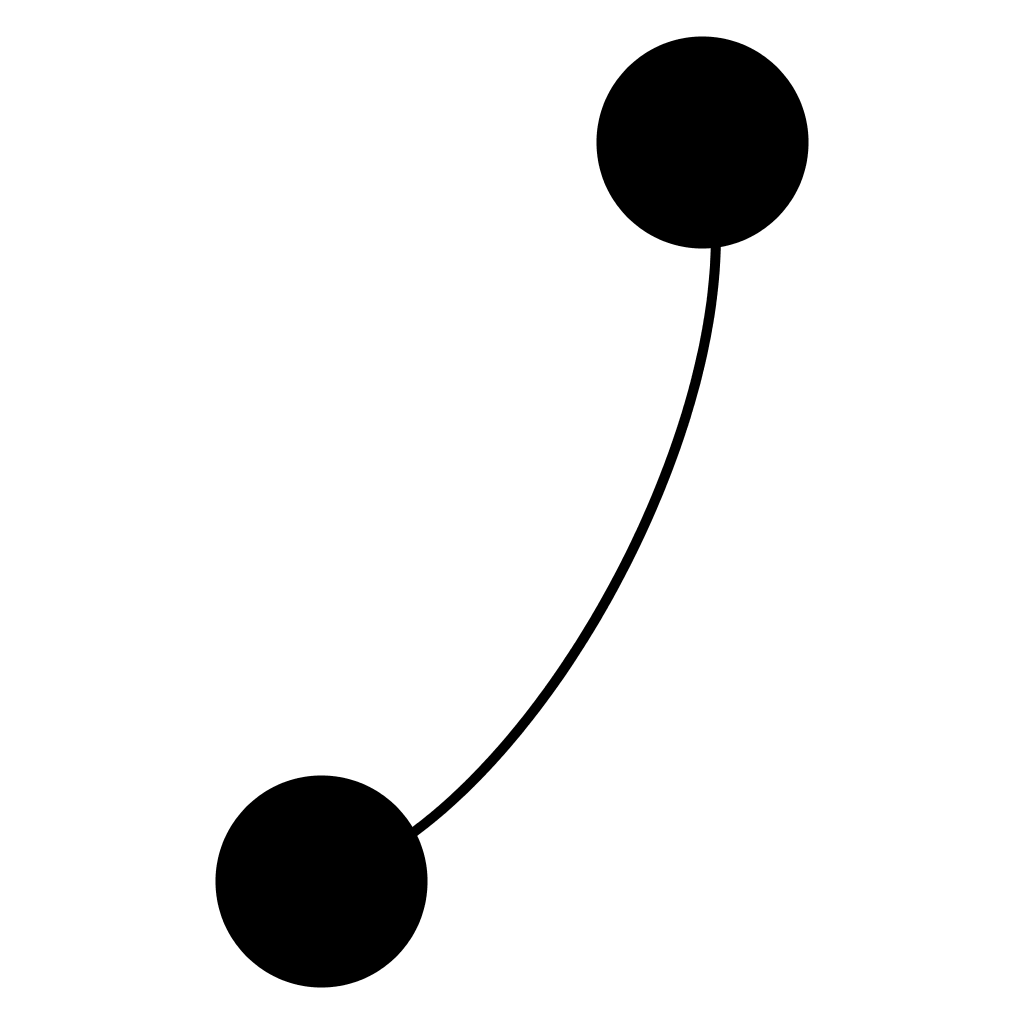}
\end{minipage}
\begin{minipage}{.12\textwidth}\centering
\includegraphics[scale=.05, angle=-50]{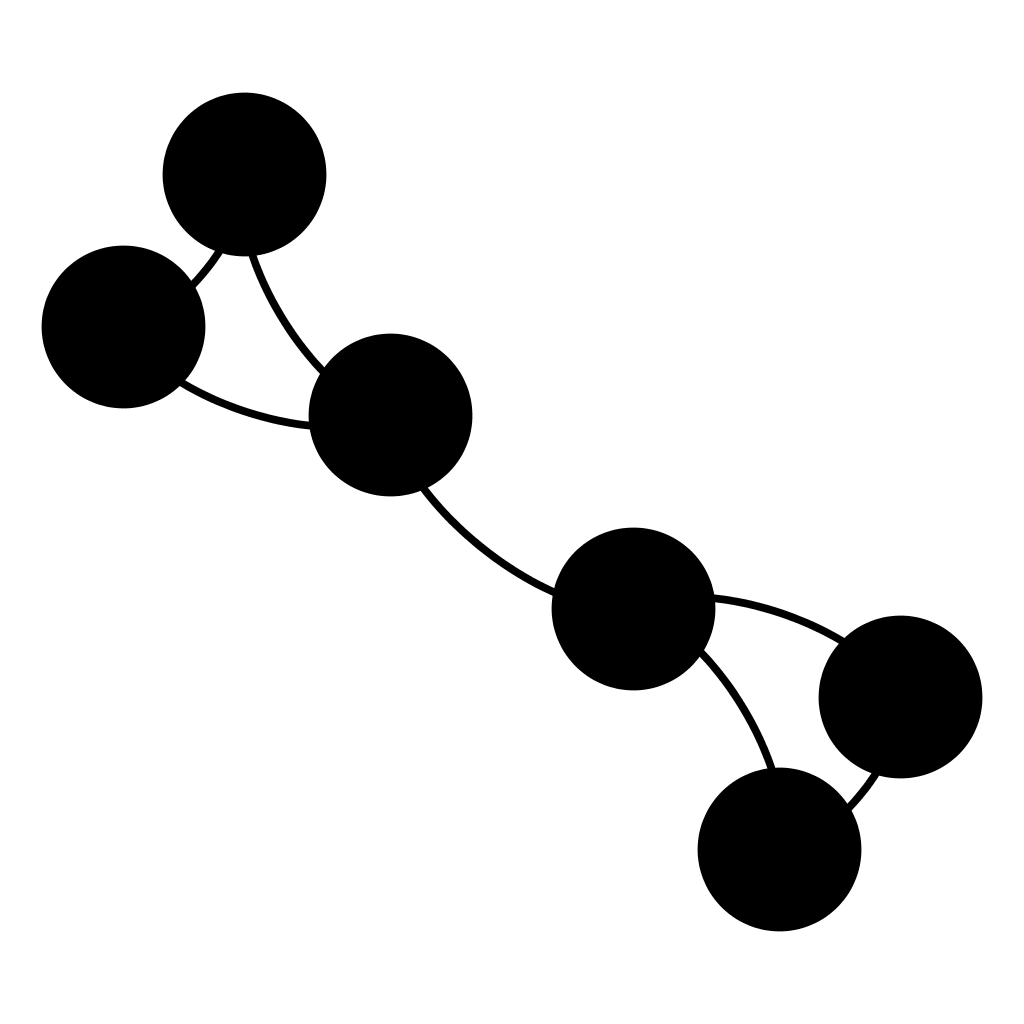}
\end{minipage}
\begin{minipage}{.34\textwidth}\centering
\includegraphics[scale=.15]{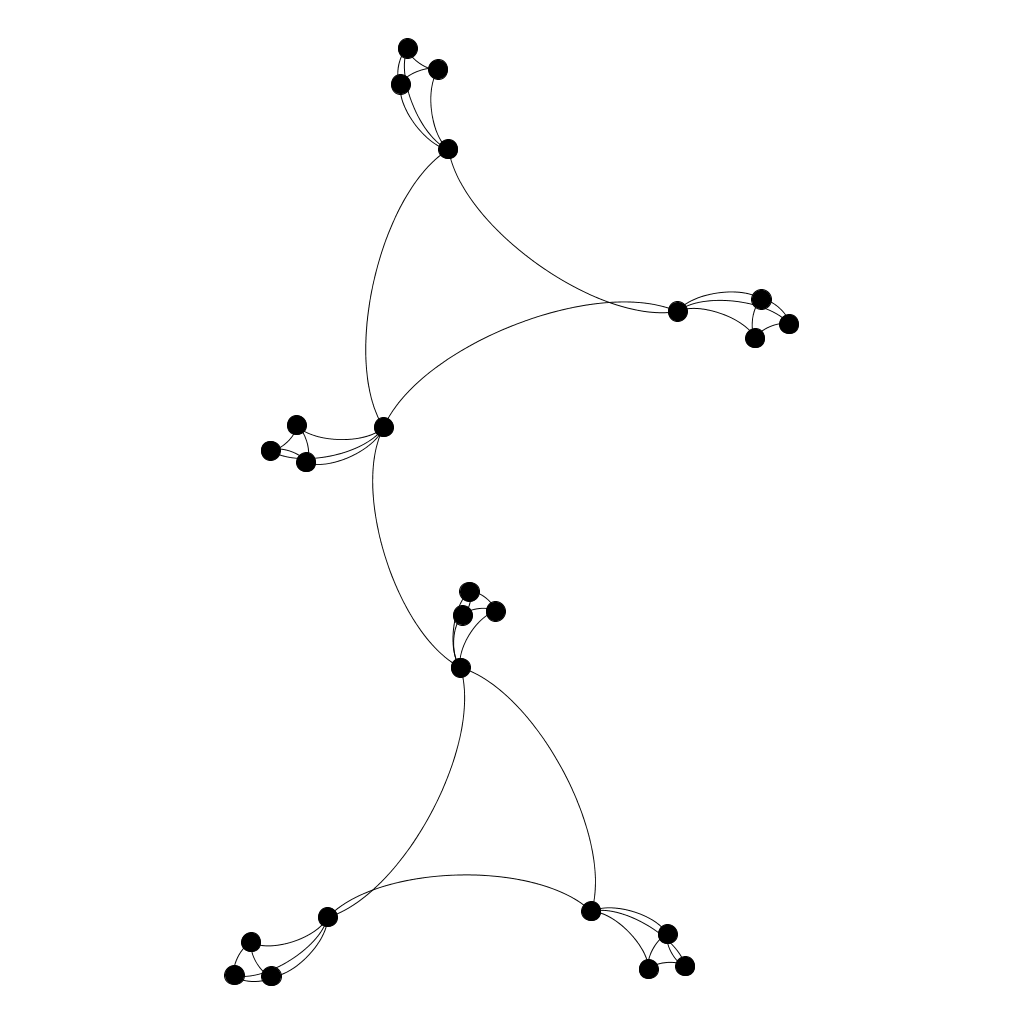}
\end{minipage}
\begin{minipage}{.37\textwidth}\centering
\includegraphics[scale=.14]{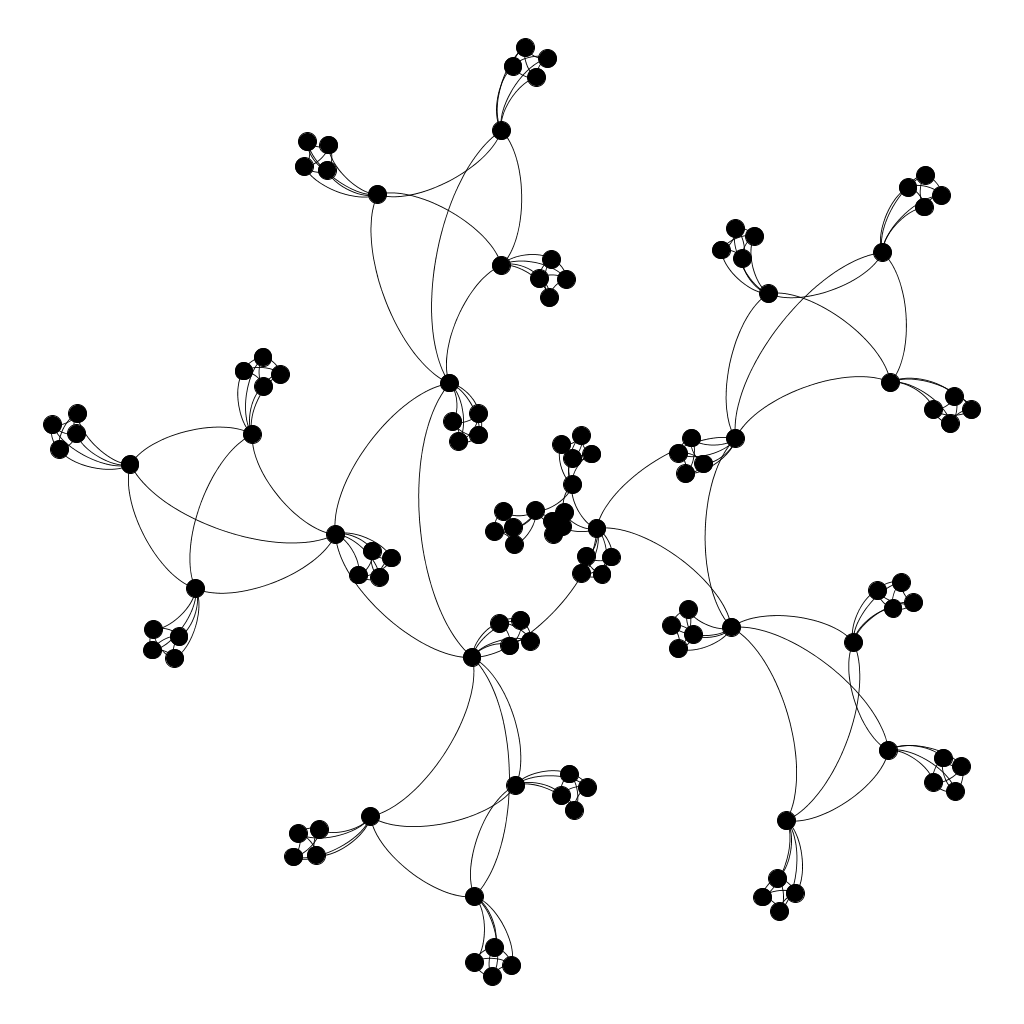}
\end{minipage}
\begin{minipage}{.45\textwidth}\centering
\includegraphics[scale=.18]{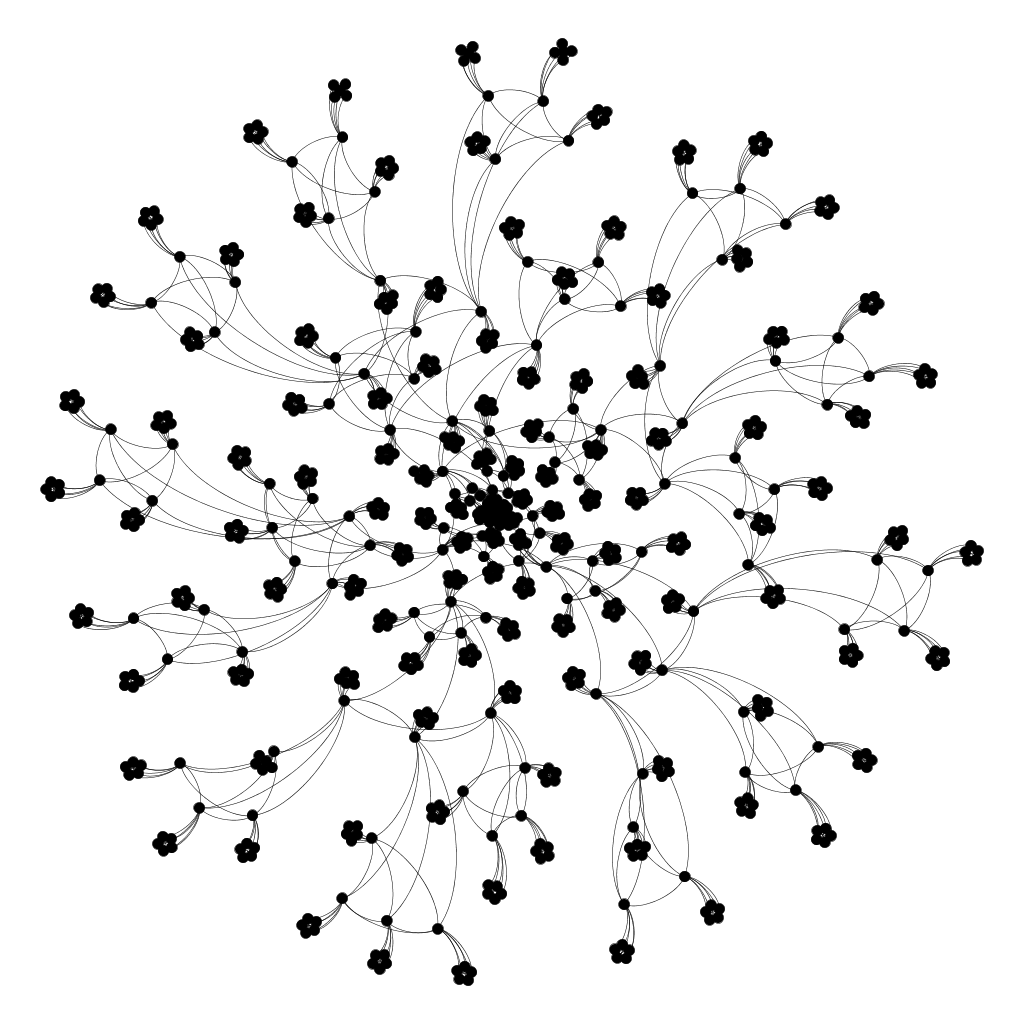}
\end{minipage}
\hfill
\begin{minipage}{.45\textwidth}\centering
\includegraphics[scale=.18]{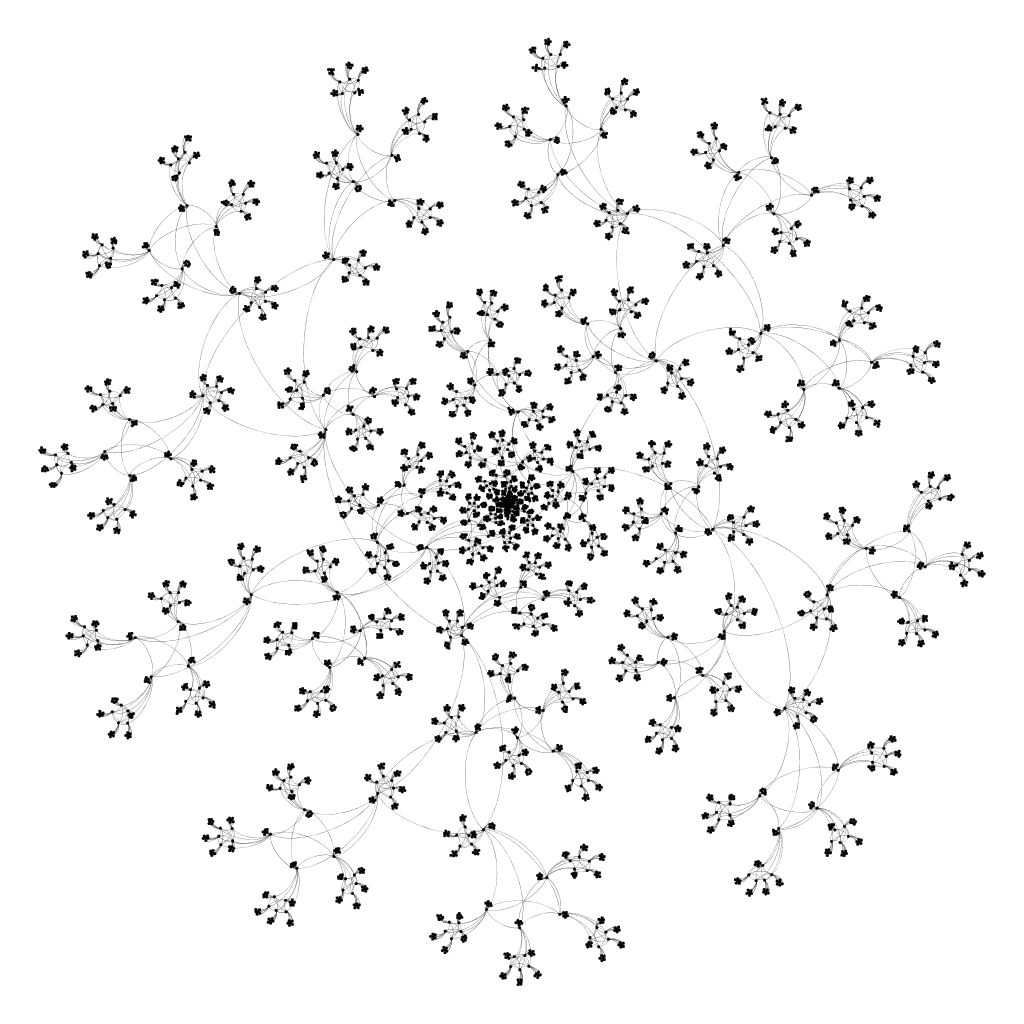}
\end{minipage}
\caption{Graphs $G_t$ generated by the cone model starting with the one-vertex graph, where $0\le t \le 6$, from top left to bottom right. In each time-step $t\ge 1$, a clique of order $t$ is added that is adjacent to each existing vertex.}\label{figcone}
\end{figure}

The frustum model was inspired in part by families of complex network models such as the Iterated Local Transitivity (ILT) model \cite{ilt} and the Iterated Local Anti-Transitivity (ILAT) model \cite{ilat}. Motivated by structural balance theory, the ILT and ILAT iteratively add transitive or anti-transitive triangles over time. Graphs generated by these models exhibit several properties observed in complex networks such as densification, small world properties, and bad spectral expansion. Both the ILT and ILAT models were unified in the recent context of Iterated Local Models in \cite{ilm}. Versions of the ILT model were considered for directed graphs \cite{directed} and hypergaphs \cite{hyper}, and a global version was considered in \cite{global}.

In the present paper, we explore the complex network and graph theoretic properties of graphs generated by the frustum model. Section~2 formally introduces the model and proves a general sufficient condition for its graphs to satisfy densification. The cone model is explored in Section~3, and it is shown that this model generates graphs which densify, satisfy the small world property, and exhibit bad spectral expansion. The final section contains several open problems and further directions.

For a general reference on graph theory, the reader is directed to the book of West \cite{west}. For background on social and complex networks, see \cite{bbook,CL}. Throughout the paper, we consider finite, undirected graphs.

\section{The frustum model}

We now formally introduce the \emph{frustum model} $\FRUSTUM(n,f,g),$ whose parameters are a positive integer $n,$ and integer-valued functions $f$ and $g.$ For simplicity, we take these functions to be non-decreasing. We will write $f_t=f(t)$ and $g_t = g(t)$.  To simplify various proofs, we assume the mild conditions that $f_t < f_{t-1}+g_{t-1}$, $f_0 \leq n$, and $n < f_0+g_0.$ 

The model generates graphs over a sets of discrete time-steps indexed by non-negative integers, with $G_0$ the clique $K_n$. Assuming $G_{t-1}$ is defined, then define $G_{t}$ as follows: for each induced clique $X$ of order $f_t$ in $G_{t-1}$, add a new set of $g_t$ vertices $Y=Y_X$ so that $X\cup Y$ forms a clique. 
Note that the newly created vertices $Y_X$ form a connected component in $G_{t} \setminus G_{t-1}$; that is, $Y_X \cap Y_{X'}=\emptyset$ for distinct $(f_{t})$-cliques $X$ and $X'$ in $G_t$.  The name of the model comes from geometric frustums, which are portions of a cone or pyramid that remain after its upper part is removed by cutting with a plane parallel to its base. We refer to the $G_t$ as \emph{frustum graphs from} $\FRUSTUM(n,f,g)$ or simply as \emph{frustum graphs}.

The \emph{cone model} is the frustum model with $f_t = 1$ for all $t$, and the \emph{cylinder model} is the frustum model with $f_t = g_t$ for all $t$. See Figure~\ref{figcyl} for an illustration of the cylinder model.

\begin{figure}[h]
\begin{minipage}{.1\textwidth}\centering
\includegraphics[scale=.007]{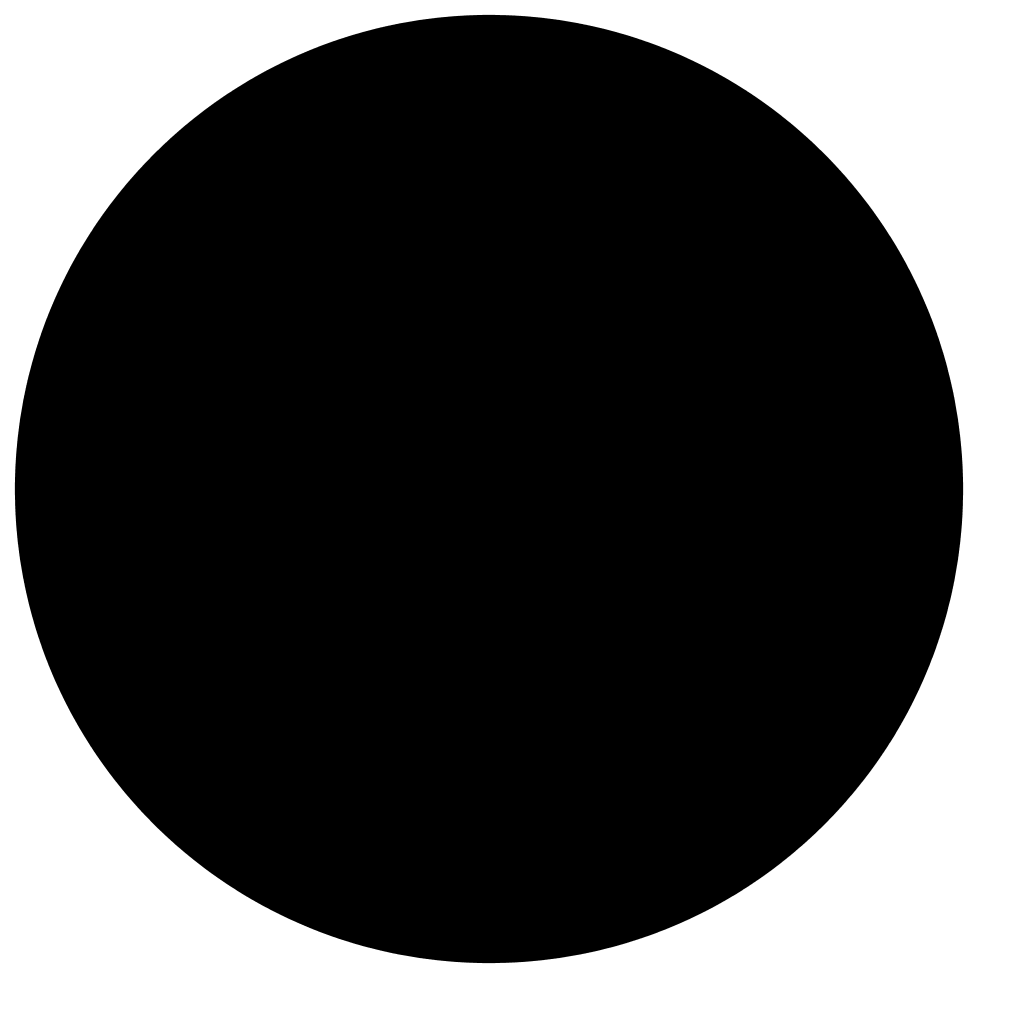}
\end{minipage}
\begin{minipage}{.1\textwidth}\centering
\includegraphics[scale=.03]{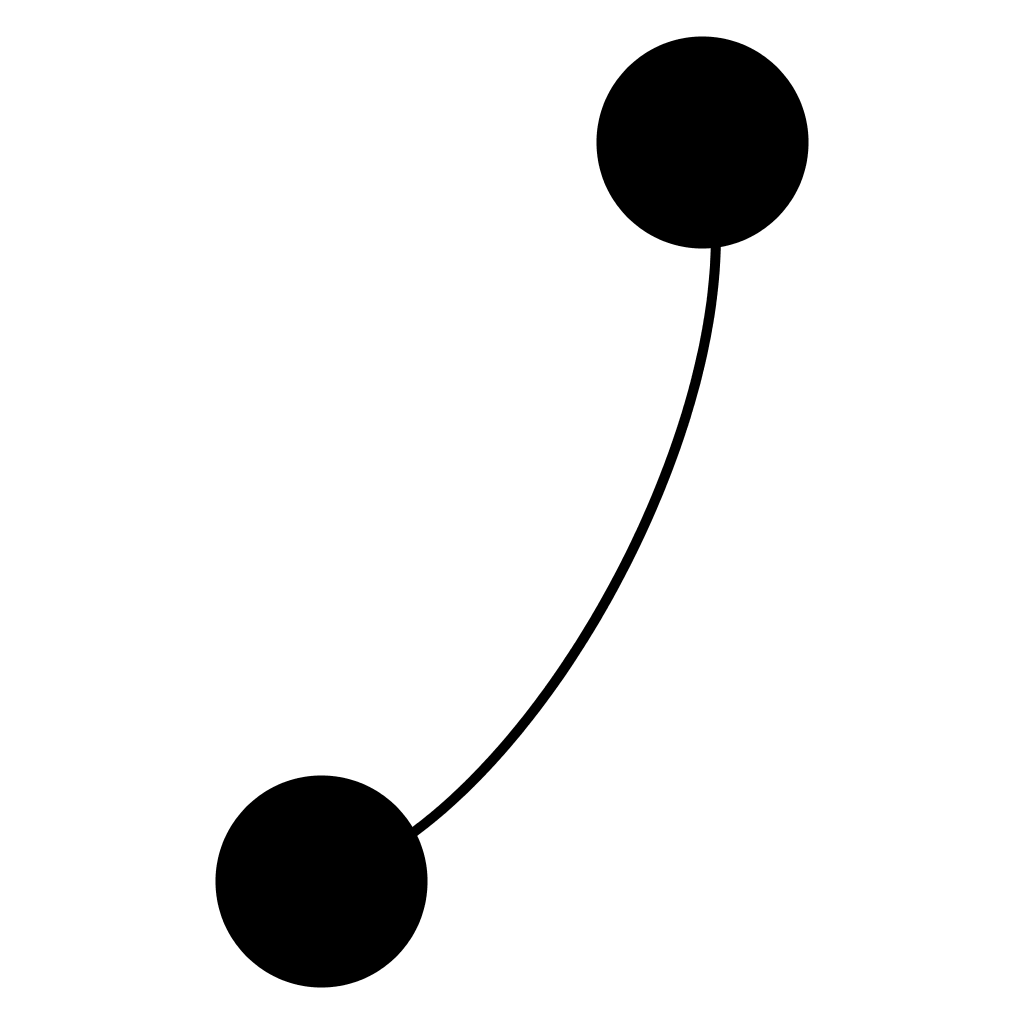}
\end{minipage}
\begin{minipage}{.1\textwidth}\centering
\includegraphics[scale=.03]{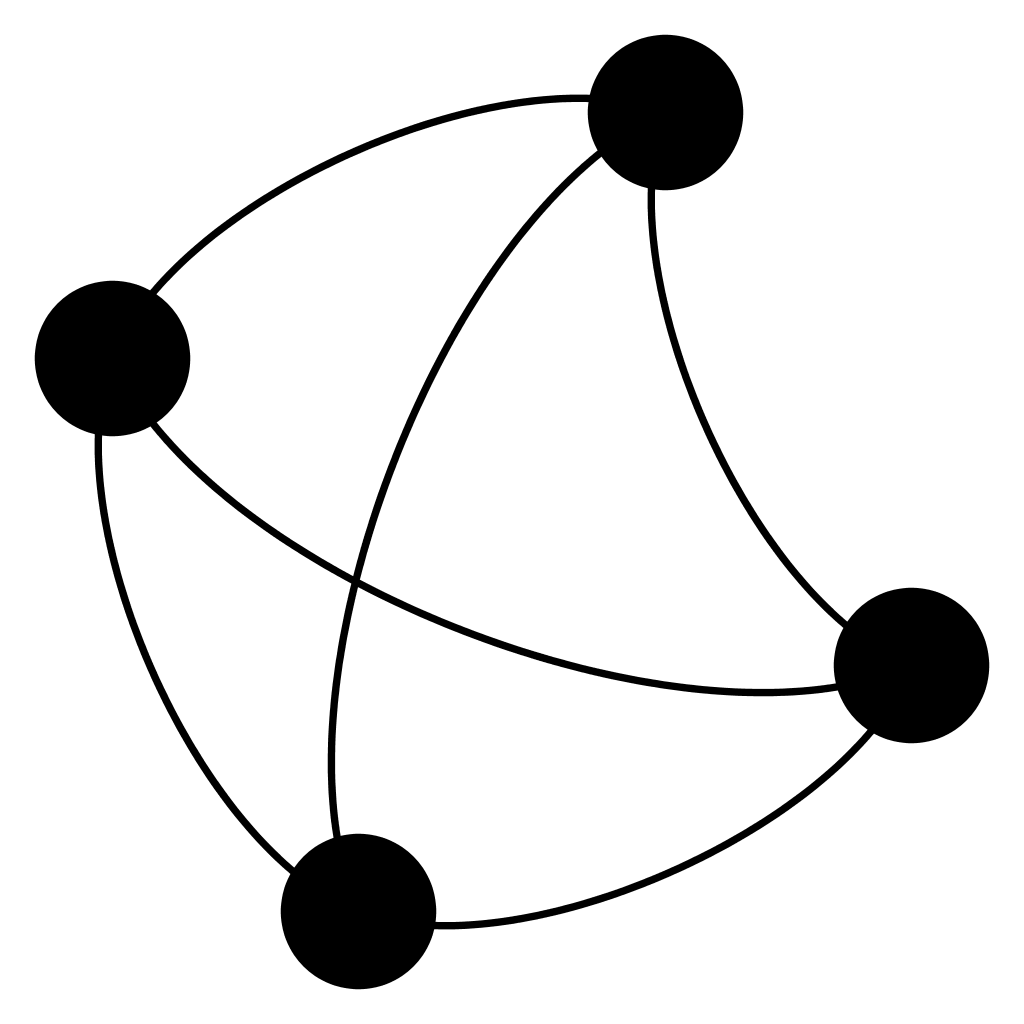}
\end{minipage}
\begin{minipage}{.3\textwidth}\centering
\includegraphics[scale=.12]{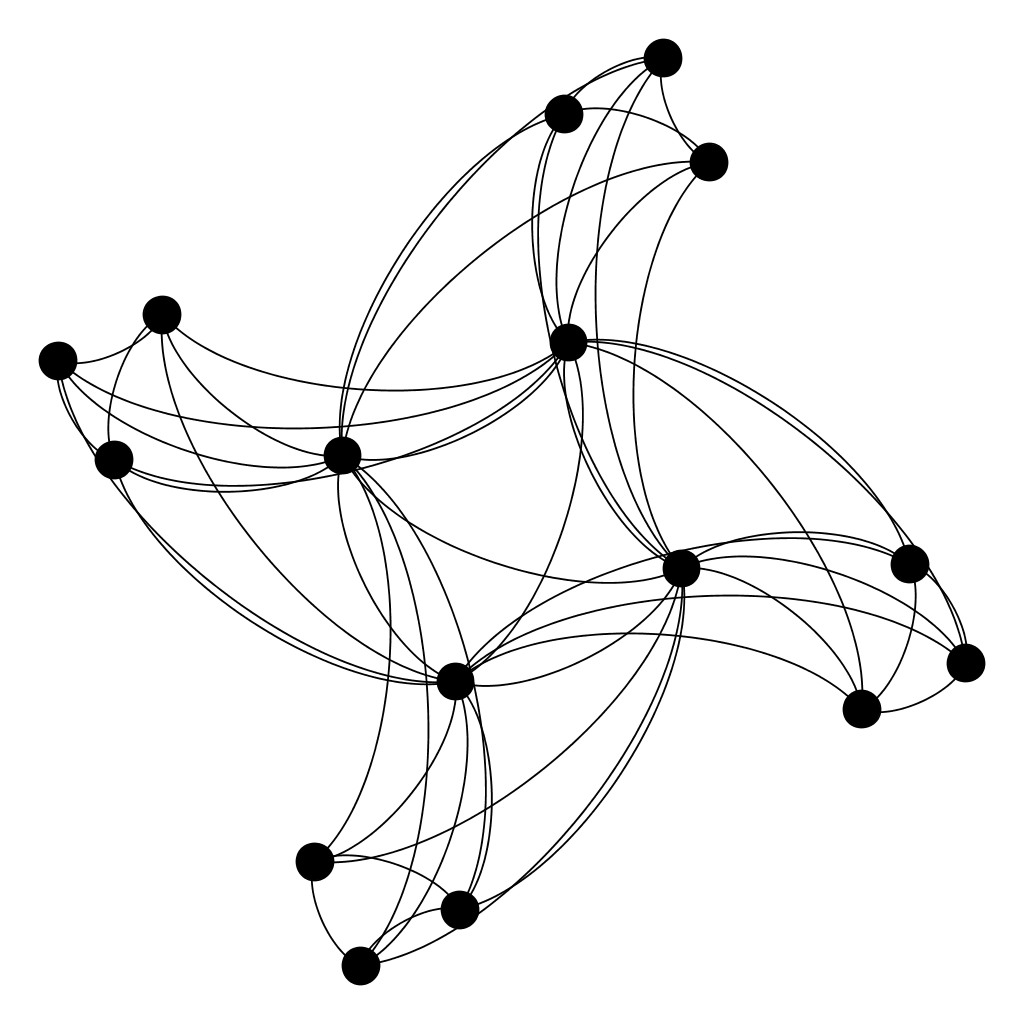}
\end{minipage}\hfill
\begin{minipage}{.3\textwidth}\centering
\includegraphics[scale=.12]{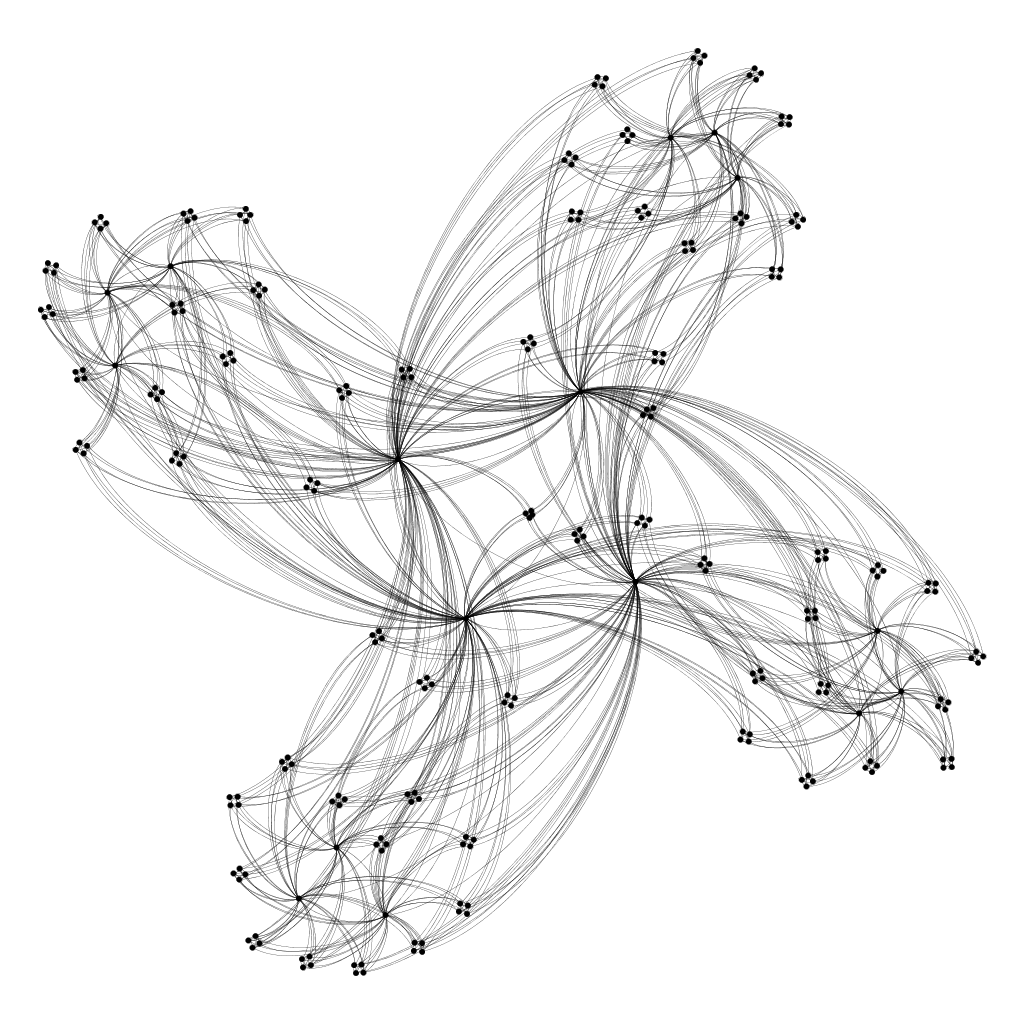}
\end{minipage}
\caption{Graphs $G_t$ of the cylinder model $\FRUSTUM(1,t,t)$ with $0\le t \le 4$.}\label{figcyl}
\end{figure}

Complex networks often exhibit densification, where the number of edges grows faster than the number of vertices; see \cite{les1}. We show densification holds for frustum graphs in a large number of choices of parameters. We give a sufficient condition for frustum model graphs to densify in our first theorem. 

Let $n_t$ be the number of vertices in $G_t$ and $e_t$ be the number of edges in $G_t$. 
Take $V_t$ to be the vertices that occur in $G_t$ but not in $G_{t-1}$. 
We define $C_t^k(u)$ to be the number of cliques of order $k$ in $G_t$ that contain the vertex $u\in V(G_t)$.

\begin{theorem} \label{lem:exp_growth}
Let $G_t$ be a frustum graph from $\FRUSTUM(n,f,g)$. If $$\liminf_{t \ge 0}{\min_{u\in V(G_{t-1})}(C_{t-1}^{f_{t}}(u))g_t / f_t}=a$$ 
with $a>1$, then $n_t$ increases as a multiple of $n_{t-1}$. Further, given a real number $\varepsilon>0$ such that $a-\varepsilon>1$, 
	\[n_t = \Omega((a-\varepsilon)^t).\]
\end{theorem}
\begin{proof}
We first show that the number of vertices created in time step $t$, $n_{t} - n_{t-1}$, is some multiple of $n_{t-1}$.  Consider the bipartite graph $B_{t}$ formed by taking the edges between vertex sets $V(G_{t-1})$ and $V_{t}$ that also occur in $E(G_{t})$. Observe that in $B_{t}$, the vertices in  $V_{t}$ all have degree $f_{t}$. In $B_{t}$, a vertex $u\in V(G_{t-1})$ has degree $C_{t-1}^{f_t}(u) g_t$. Counting the number of edges in $B_{t}$ from the perspective of each part yields $(n_{t} - n_{t-1}) f_t \geq n_{t-1} \min_{u\in V(G_{t-1})}(C_{t-1}^{f_t}(u)) g_t$. 
Thus, the number of vertices added on each iteration of this model increases multiplicatively by at least $\min_{u\in V(G_{t-1})}(C_{t-1}^{f_t}(u)) g_t / f_t$. 

For any $\varepsilon >0$, there must be a time-step $t'$ such that $$\min_{u\in V(G_{t-1})}(C_{t-1}^{f_t}(u))g_t / f_t \geq a-\varepsilon$$ 
for all $t \geq t'$.  Therefore, the number of vertices multiplicatively increases from this time-step, so the number of vertices is at least $(a-\varepsilon)^{t-t'} = \Omega((a-\varepsilon)^t)$, as required. \qed
\end{proof}

We now come to the main theorem of this section.

\begin{theorem}\label{mainden}
If $g_t + f_t = \omega(1)$ and $$\liminf_{t\ge 1}{\min_{u\in V(G_{t-1})}(C_{t-1}^{f_t}(u))g_t / f_t}>1,$$ 
then frustum graphs from $\FRUSTUM(n,f,g)$  densify in the sense that $e_t/n_t = \omega(1)$.
\end{theorem} 
\begin{proof}
By Theorem~\ref{lem:exp_growth}, for $t$ sufficiently large, we have that $n_{t}-n_{t-1} \geq (a-\varepsilon) n_{t-1}$, where $a=\liminf_{t\ge 1}{\min_{u\in V(G_{t-1})}(C_{t-1}^{f_t}(u))g_t / f_t}>1$ and $\varepsilon$ is such that $0 < \varepsilon <a-1$. Suppose for a contradiction that $e_t/n_t$ is bounded above by $b$. 

Note that as shown in Lemma \ref{lem:density} in the Appendix, $$\frac{e_{t} - e_{t-1}}{n_{t} - n_{t-1}} = \frac{g_t}{2} + f_t-\frac{1}{2}.$$ Hence, $\frac{e_{t} - e_{t-1}}{n_{t} - n_{t-1}} = \omega(1)$ since $f_t+g_t = \omega(1)$. However, we have that $$\frac{e_{t} - e_{t-1}}{n_{t} - n_{t-1}} \leq \frac{ b n_{t}} {n_{t} - n_{t-1}} = b + \frac{ b n_{t-1}} {n_{t} - n_{t-1}} \leq  b + \frac{ b n_{t-1}} {(a-\varepsilon) n_{t-1}} = b + \frac{ b } {(a-\varepsilon)},$$
 which is a contradiction. 
  \qed
\end{proof}

Theorem~\ref{mainden} applies to the cone model, since in this case, $g_t+f_t = g_t+1$ tends to infinity with $t$ and $\liminf_{t\ge 1}{\min_{u\in V(G_t)}(C_{t-1}^{f_t}(u))g_t / f_t} \geq \liminf{g_t}>1$. We have the following general result, which applies to a large set of frustum models.

\begin{corollary}\label{cor: densify frustum}
If $g_t + f_t = \omega(1)$, $f_t\geq 3$ and $g_t \geq 2$ for some $t$, and $g_{t-1}+f_{t-1} \in \{f_t-1,f_t-2\}$ for only a finite number of $t$, then the graphs generated by the $\FRUSTUM(n,f,g)$ model densify.
\end{corollary}

\begin{proof} To see this, note that $g_{t-1}+f_{t-1} \geq f_t-1$ by our assumptions on the model. Assuming $g_{t-1}+f_{t-1} \neq f_t-1$, it follows from Lemma~\ref{lem:num_cliques} in the Appendix that 
$C_{t-1}^{f_t}(u) \geq  \binom{g_{t-1}+f_{t-1}-1}{f_t-1} - \binom{f_{t-1}-1}{f_t-1} 
\geq (g_{t-1}+f_{t-1} -f_t-1)^{f_t-1}-1 \geq 2^{f_t-1}-1 \geq f_t$ when $g_{t-1}+f_{t-1} \notin \{f_t-1,f_t-2\}$. \qed
\end{proof}

For cylinder models with $f_t=g_t=t$, Corollary~\ref{cor: densify frustum} gives densification. In the case when $f_t=g_t$ is a constant, however, we observe graphs generated by the model do not densify. Consider $G_0 = K_n$ and set $f_t= n$ for all $t\geq 0$.
Denote by $C_t$ the number of cliques of order $n$ at time-step $t$. We then have that
$$
n_{t} = nC_{t-1} + n_{t-1},\quad 
e_{t} = e_{t-1} + \frac{3n^2-n}{2}C_{t-1},\quad 
\text{and } C_{t} = \binom{2n}{n}C_{t-1},
$$
where $C_0 = 1$ and $n_0 = n$. 
Note that $C_t = \binom{2n}{n}^t$. Solving the recursion, we have that 
$$n_t = n + nC_0\frac{\binom{2n}{n}^{t} - 1}{\binom{2n}{n}-1}
\quad\text{ and }\quad e_t = \binom n2 + \frac{3n^2-n}{2} C_0\frac{\binom{2n}{n}^{t} - 1}{\binom{2n}{n}-1},$$
where $C_0 =1$.
In particular, $\frac{e_t}{n_t}$ tends to the constant $\frac{3n-1}{2}$ as $t$ tends to $\infty.$

\section{Cone models}

We next take a more in-depth view of graphs from cone models, where $f_t =1$ for all $t.$ If $S$ is either a vertex of a clique in $G_{t-1},$ then we denote the newly added clique $K_{g_{t}}$ for $S$ by $\textsc{CAP}_{f,g}(S,t)$ or $\textsc{CAP}(S,t)$ when $f,g$ are clear from context. 

While we know that graphs from the cone model densify by Corollary~\ref{cor: densify frustum}, we give a more precise estimate on their densification.
\begin{theorem}
In $\FRUSTUM(1,1,g_t)$, for $t> 0$ we have that
$$n_t = \prod_{i=1}^{t} (1+g_{i})$$ and $$e_t = \sum_{i=1}^{t-1} \binom{g_{i+1}+1}{2}\prod_{j=1}^i (1+g_{j}).$$ In particular, 
${e_t}/{n_t} = \Omega(g_{t})$, and graphs generated by the model densify if $g_t=\omega(1).$
\end{theorem}
\begin{proof}
Notice that $e_0 = 0$, $n_0 = 1$ and
$$
	n_{t+1} = n_t + n_t g_{t+1} = n_t(1 + g_{t+1}).
$$
Hence, by induction we have that $n_{t+1} = \prod_{i=1}^{t+1} (1+g_{i})$. Now we also have that
\begin{align*}
	e_{t+1} &= e_{t} + n_t g_{t+1} + n_t\binom{g_{t+1}}{2} 
	= e_t + n_t\binom{g_{t+1}+1}{2}\\
	&= e_t + \binom{g_{t+1}+1}{2}\prod_{i=1}^{t} (1+g_{i})\\
	&= \sum_{i=1}^{t-1} \binom{g_{i+1}+1}{2}\prod_{j=1}^i (1+g_{j}) + \binom{g_{t+1}+1}{2}\prod_{i=1}^{t} (1+g_{i})\\
	&= \sum_{i=1}^{t} \binom{g_{i+1}+1}{2}\prod_{j=1}^i (1+g_{j}). 
\end{align*}
Further, we have that
\begin{align*}
	\frac{e_t}{n_t} &=\frac{\sum_{i=1}^{t-1} \binom{g_{i+1}+1}{2}\prod_{j=1}^i (1+g_{j})}{\prod_{i=1}^{t} (1+g_{i})} 
	\ge \frac{\binom{g_{t}+1}{2} \prod_{j=1}^{t-1} (1+g_{j})}{\prod_{i=1}^{t} (1+g_{i})} 
	= \frac{g_{t}}{2} = \Omega(g_{t}).
\end{align*}
The proof follows. \qed
\end{proof}

\subsection{Small world properties}

The following lemma gives precise values of distances of graphs from the cone model.

\begin{lemma}\label{lem:dist}
In \FRUSTUM$(1,1,g_t)$, for distinct $x,y \in V(G_{t-1})$, then we have the following.
\begin{enumerate}
\item 
$
d_{t}(x,y) = d_{t-1}(x,y);
$
\item
$d_{t}(x,y') = d_{t-1}(x,y) + 1$, where $y'\in \textsc{CAP}(y,t)$; and,

\item
$d_{t}(x',y') = d_{t-1}(x,y) + 2$, where $x'\in\textsc{CAP}(x,t)$ and $y'\in \textsc{CAP}(y,t)$.
\end{enumerate}
\end{lemma}
\begin{proof}
For the first claim, note that any shorter path would have to travel through $G_{t}$, but there are no $xy$-paths that contain vertices in $G_t$. This is because any such path would enter $G_t$ only at $\textsc{CAP}(z,t)$ for some $z \in V(G_{t-1})$ and the only way to reenter $G_{t-1}$ is to travel to $z$ again. Thus, shorter paths can only occur in $G_{t-1}$, a contradiction. 

For the second and third claim, notice that the only edge from $y'$ to $G_{t-1}$ is $yy'$ and the only paths containing $y'$ that are contained in $G_t$ is within $\textsc{CAP}(y,t)$. By the first claim, $d_{t}(x,y') = d_{t-1}(x,y) + 1$ and $d_{t}(x',y') = d_{t-1}(x,y) + 2$. \qed \end{proof}

With Lemma~\ref{lem:dist}, we prove that diameters are small in the cone model, in the sense that they grow logarithmically with their orders.

\begin{lemma}
In $\FRUSTUM(1,1,g_t)$, for $t > 0$, we have that $\diam(G_t) = 2t-1 = O(\log n_t).$
\end{lemma}
\begin{proof}
As $g$ is positive, $G_1$ is isomorphic to a clique of order $g(1)+1$. Proceeding by induction, Lemma~\ref{lem:dist} guarantees that the greatest distance between vertices will be realized for $x'\in \cone(x,t), y' \in \cone(y,t)$, where $d(x,y) = \diam(G_{t-1})$. By the induction hypothesis and Lemma~\ref{lem:dist}, $$\diam(G_t) = \diam(G_{t-1}) + 2 = 2t - 1.$$
The proof follows. \qed
\end{proof}

We study the average distances and clustering coefficient of the cone model as time
tends to infinity. Define the \emph{Wiener index} of $G_t$ as
$$
W(G_t) =\frac{1}{2} \sum_{x,y \in V (G_t)}d(x,y).$$
The Wiener index may be used to define the \emph{average distance} of $G_t$ as
$$L(G_t) = \frac{W(G_t)}{\binom{n_t}{2}}.$$

\begin{theorem}\label{thm:coneW}
In $\FRUSTUM(1,1,g_t)$, we have that
$$
W(G_t) = \prod_{i=1}^{t}(1+g_{i})\sum_{i=0}^{t-1} g_{t-i}\prod_{j=1}^{t-i-1} (1+g_{j}) \prod_{j=t-i+1}^{t} (1+g_{j}) .
$$
\end{theorem}
\begin{proof}
We claim the following:
$$
\frac{1}{2} \sum_{(x,y) \in V(G_t)^2, x \not =y} d(x,y) = \prod_{i=1}^{t}(1+g_{i})\sum_{i=0}^{t-1} g_{t-i}\prod_{j=1}^{t-i-1} (1+g_{j}) \prod_{j=t-i+1}^{t} (1+g_{j}).
$$
Notice that we may partition $V(G_t)\times V(G_t)$ into the following parts when $x\not = y$: 
For $(x,y)\in V(G_{t-1})\times V(G_{t-1})$ and $x\not = y$, we have the following.
\begin{enumerate}
	\item $\textsc{CAP}(x,t) \times \textsc{CAP}(y,t)$, which contributes $t^2(d(x,y)+2)$ to the sum;
	\item $\lbrace x\rbrace \times \textsc{CAP}(y,t)$, which contributes $t(d(x,y)+1)$ to the sum;
	\item $\textsc{CAP}(x,t) \times \lbrace y\rbrace$, which contributes $t(d(x,y)+1)$ to the sum;
	\item $(x,y)$ contributes $d(x,y)$ to the sum.
\end{enumerate}
By Lemma~\ref{lem:W} in the Appendix, we have 
$$
\frac{1}{2} \sum_{(x,y) \in V(G_t)^2, x \not =y} d(x,y) = \prod_{i=1}^{t}(1+g_{i})\sum_{i=0}^{t-1} g_{t-i}\prod_{j=1}^{t-i-1} (1+g_{j}) \prod_{j=t-i+1}^{t} (1+g_{j}).  
$$
In addition, we obtain a partition of $V(G_{t}) \times V(G_t)$. For $(x,x) \in V(G_{t-1})\times V(G_{t-1})$, we have the following.
\begin{enumerate}
	\item $\textsc{CAP}(x,t) \times \textsc{CAP}(x,t)$, which contributes $g_{t}(g_{t}-1)$ to the sum;
	\item $\lbrace x\rbrace \times \textsc{CAP}(x,t) $, which contributes $g_{t}$ to the sum;
	\item $\textsc{CAP}(x,t) \times \lbrace x\rbrace$, which contributes $g_{t}$ to the sum;
	\item $(x,x)$ contributes $0$ to the sum.
\end{enumerate}	
Hence, we have that
$$
\frac{1}{2} \sum_{(x,x) \in V(G_t)\times V(G_t)} d(x,y) = \frac{1}2 \sum_{(x,x) \in V(G_{t-1})\times V(G_{t-1})} (g_{t}(g_{t}-1)+2f_t) = \frac{1}2 g_{t}\prod_{i=1}^t (1+g_{i}).
$$
Thus, 
$$
W(G_t) = \prod_{i=1}^{t}(1+g_{i})\sum_{i=0}^{t-1} g_{t-i}\prod_{j=1}^{t-i-1} (1+g_{j}) \prod_{j=t-i+1}^{t} (1+g_{j}) + \frac{1}2 g_{t}\prod_{i=1}^t (1+g_{i}).
$$
The proof follows. \qed
\end{proof}

As a consequence, we derive upper bounds on the average distance in cone models.

\begin{corollary}
In $\FRUSTUM(1,1,g_t)$ for $t > 0$,
$$L(G_t) = \Theta\left(g_{t} + t - \sum_{i=1}^{t} \frac{1}{g_{i}+1}\right).$$
Hence, $L(G_t) = O(\log n_t)$ if $g_{t} = O(t)$ and $L(G_t) = O(g_{t})$ if $g_{t} = \omega(t)$.
\end{corollary}
\begin{proof}
We have
\begin{align*}
L(G_t) &= \frac{ \prod_{i=1}^{t}(1+g_{i})\sum_{i=0}^{t-1} g_{t-i}\prod_{j=1}^{t-i-1} (1+g_{j}) \prod_{j=t-i+1}^{t} (1+g_{j})+ \frac{1}2 g_{t}\prod_{i=1}^t (1+g_{i})}{\binom{n_t}{2}}\\
&= \frac{2\sum_{i=0}^{t-1} g_{t-i}\prod_{j=1}^{t-i-1} (1+g_{j}) \prod_{j=t-i+1}^{t} (1+g_{j})+ g_{t}\prod_{i=1}^t (1+g_{i})}{\prod_{i=1}^{t} (1+g_{i}) - 1}.
\end{align*}
Now notice that 
$$ 
L(G_t) = \Theta\left(\frac{\sum_{i=0}^{t-1} g_{t-i}\prod_{j=1}^{t-i-1} (1+g_{j}) \prod_{j=t-i+1}^{t} (1+g_{j})}{\prod_{i=1}^{t} (1+g_{i})} + g_{t} \right)
$$
and 
\begin{align*}
\frac{\sum_{i=0}^{t-1} g_{t-i}\prod_{j=1}^{t-i-1} (1+g_{j}) \prod_{j=t-i+1}^{t} (1+g_{j})}{\prod_{i=1}^{t} (1+g_{i})}
&= \sum_{i=0}^{t-1}  \frac{g_{t-i}}{1+g_{t-i}}\\
&= t - \sum_{i=1}^{t} \frac{1}{g_{i}+1}.
\end{align*}
The proof follows. \qed
\end{proof}

Complex networks often exhibit high clustering, as measured by their clustering
coefficients; see \cite{bbook}. Informally, clustering measures local density. For a vertex $x$ of $G,$ let $e(x)$ be the number of edges in the subgraph induced by the neighbors of $x.$ The \textit{clustering coefficient} of $G$ is defined by $C(G) = \frac{1}{|V(G)|} \sum_{x \in V(G)} c_x (G),$ where $c_x (G) = \frac{e(x)}{{\deg (x) \choose 2}}.$ 
 Unlike the PA or ILT models where the clustering coefficient tends to 0 with $t$ \cite{ilt}, frustum model graphs have high clustering with clustering coefficients bounded away from 0.

\begin{lemma}\label{lem:deg1}
In $\FRUSTUM(1,1,g_t)$ for $t >0$ and $1 \le j \le t$, $x \in V(G_j) \setminus V(G_{j-1})$, we have that $\deg_t(x) = \sum_{i=j}^{t} g_{i}.$

\end{lemma}
\begin{proof}
Note that $G_1$ is a clique of order $g(1)+1$. For $t > 1$ and $1 \le j \le t$, $x \in V(G_j) \setminus V(G_{j-1})$, the degree of $x$ increases at time $t$ by $g_{t}$. Hence, 
$$
\deg_t(x) = \deg_{t-1}(x) + g_{t} = \sum_{i=j}^{t-1} g_{i} + g_{t} = \sum_{i=j}^{t} g_{i}.
$$
We say that $\deg_{t-1}(x) = 0$ if $x \in V(G_t) \setminus V(G_{t-1})$. The proof follows. \qed 
\end{proof}

A lower bound on the clustering coefficient of graphs from the cone models is given in the following theorem.

\begin{theorem}
In $\FRUSTUM(1,1,g_t)$, we have that $C(G_t) = \Theta(1).$
\end{theorem}
\begin{proof}
For a vertex $x$, define $e(x,t)$ to be $e(x)$ in $G_t$. For $1 \le j \le t$, $x \in V(G_j) \setminus V(G_{j-1})$, we have that
$$
e(x,t) = \sum_{i=j}^{t} \binom{g_{i}}{2}.
$$
Let $c_t(x)$ be the clustering coefficient of $x$ at time $t.$ Observe that 
$$
c_t(x) = \frac{\sum_{i=j}^t \binom{g_{i}}{2}}{\binom{\sum_{i=j}^{t} g_{i}}{2}} = \Theta\left(\frac{\sum_{i=j}^t g_{i}^2}{\left(\sum_{i=j}^{t} g_{i}\right)^2}\right).
$$
In particular,
$$
 \frac{1}2\frac{\sum_{i=j}^tg_{i}^2}{\left(\sum_{i=j}^{t} g_{i}\right)^2} \le c_t(x) \le 2\frac{\sum_{i=j}^tg_{i}^2}{\left(\sum_{i=j}^{t} g_{i}\right)^2}.
$$
By the Cauchy-Schwarz inequality, we have that
$$
\left(\sum_{i=j}^{t} g_{i}\right)^2 \le (t-j + 1)\sum_{i=j}^t g_{i}^2  
$$
and so
$$
\frac{\sum_{i=j}^t g_{i}^2}{\left(\sum_{i=j}^{t} g_{i}\right)^2 } \ge \frac{1}{t-j+1}.
$$
We then derive that
\begin{align*}
\sum_{x\in V(G_t)} c_t(x) 
	&= \sum_{j = 1}^t \sum_{x\in V(G_j)\setminus V(G_{j-1})} c_t(x) 
	\ge \frac12\sum_{j = 1}^t \sum_{x\in V(G_j)\setminus V(G_{j-1})} \frac{1}{t-j+1}\\
	&= \frac12\sum_{j = 1}^t \frac{n_j - n_{j-1}}{t-j+1}
	= \frac12\sum_{j = 1}^t \frac{g_{j} \prod_{i=1}^{j-1} (1+g_{i})}{t-j+1} 
    \ge \frac12 g_{t} n_{t-1}.
\end{align*}
Hence,
$$
C(G_t) \ge \frac12\left(\frac{g_{t}}{1+g_{t}}\right) = \Omega(1).
$$
The proof follows. \qed
\end{proof}

\subsection{Spectral expansion}

For a graph $G$ and sets of vertices $X,Y \subseteq V(G)$, define $E(X,Y)$ to be the set of edges in $G$ with one endpoint in $X$ and the other in $Y.$ For simplicity, we write $E(X)=E(X,X).$ Let $A$ denote the adjacency matrix and $D$ denote the diagonal degree matrix of a graph $G$. The \emph{normalized Laplacian} of $G$ is
\[ \mathcal{L} = I - D^{-1/2}AD^{-1/2}.\]
Let $0 = \lambda_0 \leq \lambda_1 \leq \cdots \leq \lambda_{n-1} \leq 2$ denote
the eigenvalues of $\mathcal{L}$. The \emph{spectral gap} of the normalized Laplacian is defined as
\[
\lambda = \max\{ |\lambda_1 - 1|, |\lambda_{n-1} - 1| \}.
\]

We will use the expander mixing lemma for the normalized Laplacian~\cite{sgt}. For sets of vertices $X$ and $Y$, we use the notation $\vol(X) = \sum_{v \in X} \deg(v)$ for the volume of $X$, $\overline{X} = V \setminus X$ for the complement of $X$, and, $e(X,Y)$ for the number of edges with one end in each of $X$ and $Y.$ Note that $X \cap Y$ need not be empty, and in this case, the edges completely contained in $X\cap Y$ are counted twice. In particular, $e(X,X) = 2 |E(X)|$.

\begin{lemma}[Expander mixing lemma]\cite{sgt}\label{mix}
If $G$ is a graph with spectral gap $\lambda$, then, for all sets $X \subseteq V(G),$
\[
\left| e(X,X) - \frac{(\vol(X))^{2}}{\vol(G)} \right| \leq \lambda \frac{\vol(X)\vol(\overline{X})}{\vol(G)}.
\]
\end{lemma}

A spectral gap bounded away from zero is an indication of bad expansion properties, which is characteristic for social networks; see \cite{estrada}. The next theorem represents a drastic departure for graphs from cone models from the good expansion found in binomial random graphs, where $\lambda = o(1)$~\cite{sgt,CL}.

\begin{theorem}
In $\FRUSTUM(1,1,g_t)$, we have that $\lambda_t \ge \frac{1}{2}.$
\end{theorem}
\begin{proof}
Let $X = V(G_t)\setminus V(G_{t-1})$ be the set of new vertices used to create $G_t$ from $G_{t-1}$.  The following equations hold:
\begin{align*}
\vol(X)&= (n_t - n_{t-1})g_{t} = n_{t-1} g_{t}^2,\\
\vol(G_t)&=2 e_t = 2e_{t-1} + 2 \binom{g_{t}+1}{2}n_{t-1}\\
&= 2e_{t-1} + (g_{t}^2+g_{t})n_{t-1},\\
\vol(\overline{X})&=\vol(G_t) - \vol(X) \\
&= 2e_{t-1} + (g_{t}^2+g_{t})n_{t-1}-n_{t-1}g_{t}^2\\
&= 2e_{t-1} + n_{t-1}g_{t},\\
e(X,X) &= 2\binom{g_{t}}{2}n_{t-1} = g_{t}(g_{t}-1)n_{t-1}.
\end{align*}
Therefore, we derive that
\begin{align*}
e(X,X) - \frac{\vol(X)^2}{\vol(G_t)} &= 
\frac{n_{t-1}g_{t}(2(g_{t}-1)e_{t-1}-n_{t-1}g_{t})}
{2e_{t-1} + (g_{t}^2+g_{t})n_{t-1}}
\end{align*}
\begin{align*}
\frac{\vol(\overline{X})\vol(X)}{\vol(G_t)} &= 
\frac
{
n_{t-1}g_{t}^2(2e_{t-1}+n_{t-1}g_{t})
}
{
2e_{t-1} + (g_{t}^2+g_{t})n_{t-1}.
}
\end{align*}
Hence, by the Lemma~\ref{mix}, we have that
\begin{align*}
\lambda \ge 
\frac
{
|2(g_{t}-1)e_{t-1}-n_{t-1}g_{t}|
}
{
g_{t}(2e_{t-1}+n_{t-1}g_{t})
}
\ge \frac{2e_{t-1}}{2e_{t-1}+n_{t-1}g_{t}}
= \frac{2}{2 + g_{t}n_{t-1}/e_{t-1} }
\ge \frac{2}{2 + 2 } = \frac12.
\end{align*}
We use here that $n_t = o(e_t)$ and ${n_{t-1}}/{e_{t-1}} \le {2}/{g_{t}}$. \qed \end{proof}

\section{Conclusion and future directions}

We introduced the frustum model for complex networks, which is a deterministic model formed by iteratively extending cliques with parametrized orders over discrete time-steps. For a wide range of parameters, the frustum graphs densify over time. In the case of the cone model where one vertex cliques are extended, we showed that the model generates small world graphs with small distances and high clustering coefficients. We also showed that graphs from the cone model exhibit bad spectral expansion with respect to their normalized Laplacian matrices.

Many directions remain unexplored and will be considered in the full version of the paper. Several of the results for the cone model should go through if we assume $f_t=2,$ where we extend edges rather than vertices; see Figure~\ref{figfr}. Another interesting case to consider is when $g_t=1$, which is akin to iteratively adding inverted cones; see Figure~\ref{figfr}.  

\begin{figure}[h]
\centering
\begin{subfigure}
  \centering
  \includegraphics[width=.45\linewidth]{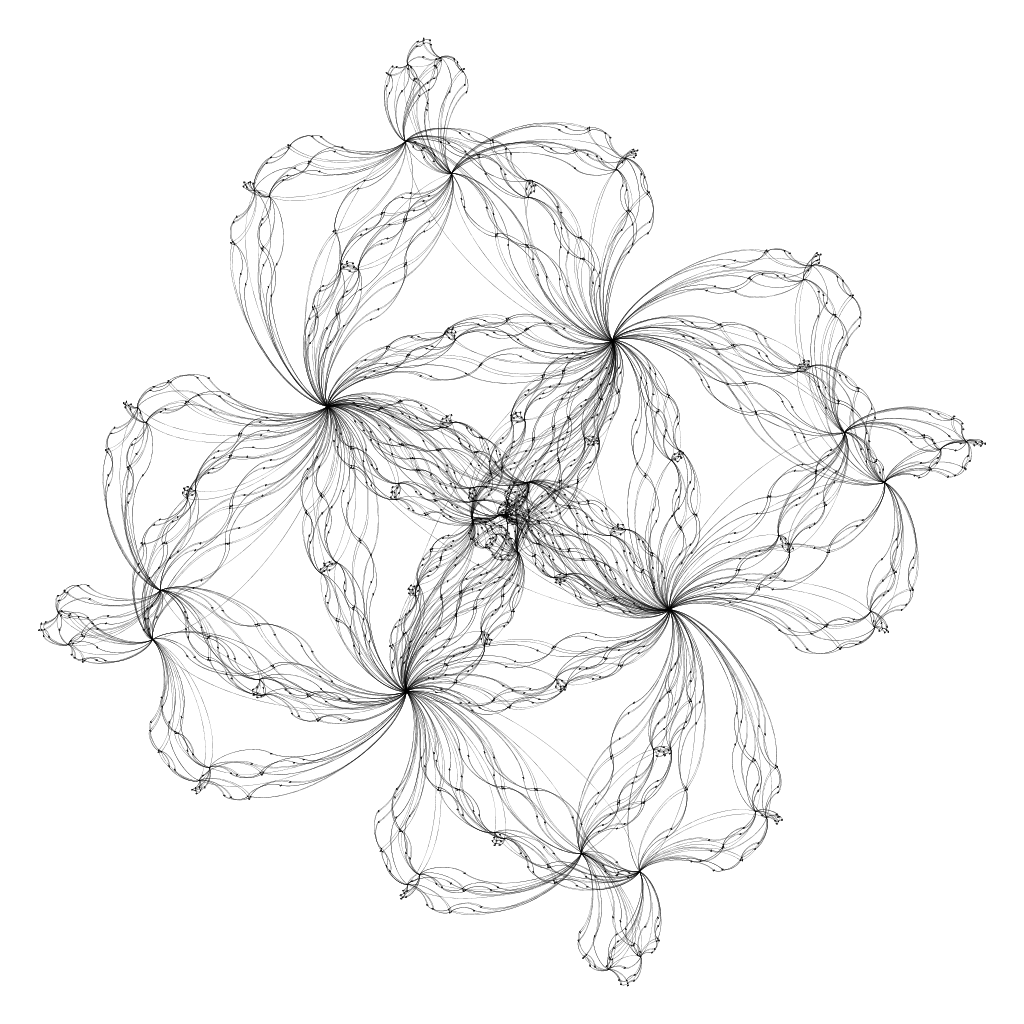}
  \label{fig:sub1}
\end{subfigure}%
\begin{subfigure}
  \centering
  \includegraphics[width=.45\linewidth]{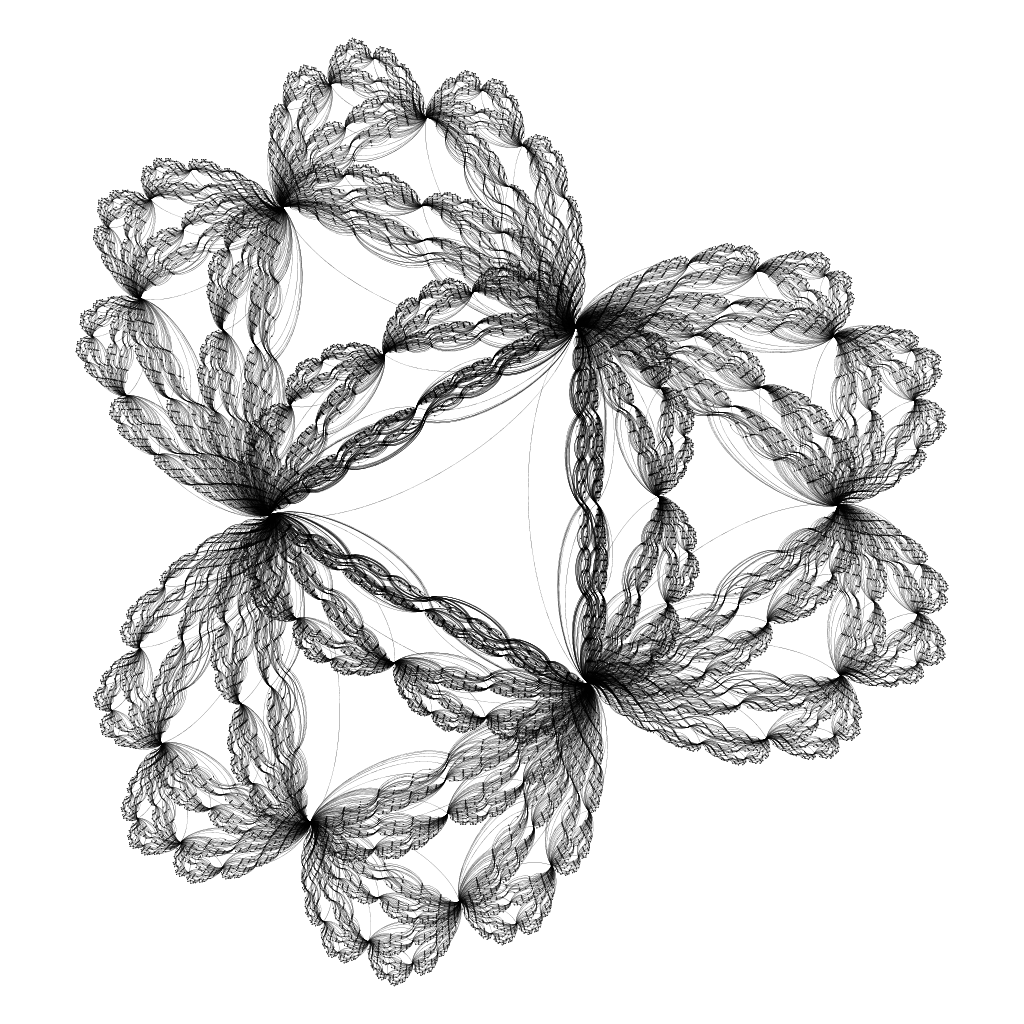}
  \label{fig:sub2}
\end{subfigure}
\caption{Frustum graphs when $f_t=2$, with $g_t = 2$ on the left at $t=5$, and $g_t=1$ on the right at $t=10$.  \label{figfr}}
\end{figure}

Finding a necessary and sufficient condition for densification based on the parameters $f$ and $g$ remains open. A natural question is to explore spectral expansion for models other than the cone model. Stochastic variations of the frustum model are of interest, where the order of the cliques added in each time-step is controlled by a random variable. An interesting direction is to explore graph theoretic parameters of frustum graphs such as their domination, chromatic, and cop numbers.

\section*{Appendix}

We present technical results from the paper not included in the main text due to space considerations. We begin with two combinatorial lemmas on the general frustum model.

\begin{lemma} \label{lem:density}
In the frustrum model $\FRUSTUM(n,f,g)$ with $e_t$ edges and $v_t$ vertices at time $t \ge 1$, we have that
$$\frac{e_{t} - e_{t-1}}{n_{t} - n_{t-1}} = \frac{g_t}{2} + f_t-\frac{1}{2}.$$
\end{lemma}
\begin{proof}
Let $C^{k}_{t}$ be the total number of cliques of order $k$ in $G_t$. 
During the $t$th time step, each of the $C^{f_t}_{t-1}$ cliques of order $f_t$ were extended to $(f_t + g_t)$-cliques, and each such extension led to $\binom{f_t+g_t}{2} - \binom{f_t}{2}$ new edges and $g_t$ new vertices. 
Thus, the number of edges created during time $t$, which is $e_{t} - e_{t-1}$, is $\left(\binom{g_t + f_t}{2} - \binom{f_t}{2}\right)C^{f_t}_{t-1}$. The number of vertices created during time $t$, which is $n_t - n_{t-1}$, is $g_t C^{f_t}_{t-1}$. The proof follows.  \qed
\end{proof}

\begin{lemma} \label{lem:num_cliques}
In the frustrum model $\FRUSTUM(n,f,g)$ at time $t-1$ and with $u \in V(G_{t-1})$, we have that
\[C_{t-1}^{f_t}(u) \geq  \binom{g_{t-1}+f_{t-1}-1}{f_t-1} - \binom{f_{t-1}-1}{f_t-1}.\]
\end{lemma}
\begin{proof} 
The result holds by definition if $u \in V_{t-1}$, so we assume this is not the case. We analyze the number of cliques of order $k$ in $G_{t-1}$ that include a vertex $u \in V(G_{t-2})$. 
There are $C_{t-2}^{k}(u)$ cliques of order $k$ in $G_{t-2}$ that include $u$.
When applying the model at stage $t-1$ to obtain the graph $G_{t-1}$, each clique of $G_{t-1}$ of order $g_{t-1}+f_{t-1}$ that includes $u$ and a vertex in $V_{t-1}$ must have been created from one of the $C_{t-2}^{f_{t-1}}(u)$ cliques of order $f_{t-1}$ in $G_{t-2}$. 
Each of these $(g_{t-1}+f_{t-1})$-cliques will contain $\binom{g_{t-1}+f_{t-1}-1}{k-1} - \binom{f_{t-1}-1}{k-1}$ cliques of order $k$ that both contain $u$ and at least one vertex in $V_t$. 

We therefore have that
\[
C_{t-1}^{k}(u) = 
C_{t-2}^{k}(u) + 
\left(\binom{g_{t-1}+f_{t-1}-1}{k-1} - \binom{f_{t-1}-1}{k-1}\right)C_{t-2}^{f_{t-1}}(u).
\]
Noting that $C_{t-2}^{f_{t-1}}(u)\geq 1$ by the mild assumptions made about the frustrum model in its definition, this yields that $C_{t-1}^{k}(u) \geq \binom{g_{t-1}+f_{t-1}-1}{k-1} - \binom{f_{t-1}-1}{k-1}$. Substituting in $k=f_t$ finishes the proof.  \qed   
\end{proof}

The following lemma is used in our analysis of the cone model.

\begin{lemma}\label{lem:W}
In $\FRUSTUM(1,1,g_t)$, we have
$$\frac{1}{2} \sum_{(x,y) \in V(G_t)^2, x \not =y} d(x,y) = \prod_{i=1}^{t}(1+g_{i})\sum_{i=0}^{t-1} g_{t-i}\prod_{j=1}^{t-i-1} (1+g_{j}) \prod_{j=t-i+1}^{t} (1+g_{j}).  
$$
\end{lemma}
\begin{proof}
Using the partition of $V(G_t)\times V(G_t)$, $x\not = y$, given in Theorem~\ref{thm:coneW},
we may proceed by induction to obtain that:
\begin{align*}
&\frac{1}{2} \sum_{(x,y) \in V(G_t)^2, x \not =y} d(x,y) \\
&= \frac{1}{2} \sum_{x,y \in V(G_{t-1}), x\not=y} \Big( 
	\sum_{(w,z) \in \textsc{CAP}(x,t)\times \textsc{CAP}(y,t)}  d(w,z) \\
	&\qquad+ \sum_{(x,z) \in \lbrace x\rbrace\times \textsc{CAP}(y,t)}  d(x,z) \\
	&\qquad+ \sum_{(z,y) \in \textsc{CAP}(x,t) \times \lbrace y\rbrace}  d(z,y)\\
	&\qquad+ d(x,y) \Big)\\
&= \frac{1}{2} \sum_{x,y \in V(G_{t-1}), x\not=y} \Big( g_{t}^2 (d(x,y)+2)+2f_t(d(x,y)+1)+d(x,y)\Big)\\	
&= \frac{1}{2} \sum_{x,y \in V(G_{t-1}), x\not=y} \Big( (g_{t}+1)^2 d(x,y) + 2f_t(g_{t}+1) \Big)\\
&= \frac{1}{2} \sum_{x,y \in V(G_{t-1}), x\not=y} (g_{t}+1)^2 + \prod_{i=1}^{t} (1+g_{i}) \prod_{i=1}^{t-1} (1+g_{i}) g_{t}\\
&= \left(\prod_{i=1}^{t-1}(1+g_{i})\sum_{i=0}^{t-2} g_{t-1-i}\prod_{j=1}^{t-i-2} (1+g_{j}) \prod_{j=t-i}^{t-1}  (1+g_{j})\right) (g_{t}+1)^2 \\
	&\qquad+ \prod_{i=1}^{t} (1+g_{i}) \prod_{i=1}^{t-1} (1+g_{i}) g_{t}\\
&= \prod_{i=1}^{t}(1+g_{i})\sum_{i=1}^{t-1} g_{t-i}\prod_{j=1}^{t-i-1} (1+g_{j}) \prod_{j=t-i+1}^{t} (1+g_{j})  \\
	&\qquad+ \prod_{i=1}^{t} (1+g_{i}) \prod_{i=1}^{t-1} (1+g_{i}) g_{t}\\	
&= \prod_{i=1}^{t}(1+g_{i})\sum_{i=0}^{t-1} g_{t-i}\prod_{j=1}^{t-i-1} (1+g_{j}) \prod_{j=t-i+1}^{t} (1+g_{j}).  \\	
\end{align*}

The proof follows. \qed
\end{proof}

\end{document}